\documentclass[12pt,notitlepage]{article}
\usepackage{amsmath,amsthm,amssymb,amsfonts}
\usepackage[top=1.5in,bottom=1in,margin=1.5in]{geometry}
\usepackage{color}

\definecolor{refkey}{rgb}{0,0.08,0.65}
\definecolor{labelkey}{rgb}{0.8,0,0}

\newcommand{\res}{\mathrm{Res}}

\newcommand{\pa}{\partial}

\newcommand{\bt}{\mathbf{t}}

\newcommand{\by}{\mathbf{y}}
\newcommand{\btt}{\mathbf{t'}}
\newcommand{\ld}{\lambda}
\newcommand{\ta}{\tau}

\newcommand{\sg}{\sigma}

\newcommand{\defeq}{\stackrel{\mathrm{def}}{=}}

\newtheorem{lm}{Lemma}

\newtheorem{rmk}{Remark}
\newtheorem{thm}{Theorem}
\newtheorem{prop}{Proposition}
\newtheorem{eg}{Example}

\begin{document}

\begin{center}
\section*{Bilinear Identities and Hirota's Bilinear Forms for an Extended Kadomtsev-Petviashvili
  Hierarchy}

\noindent
{Runliang Lin$^1$, Xiaojun Liu$^2$ and Yunbo Zeng$^1$  }

\end{center}

\noindent
{\small {\it $^1$Department of Mathematical Sciences, Tsinghua
University, Beijing 100084, P.R. China,
e-mail: }\texttt{rlin@math.tsinghua.edu.cn}{\it ;} \texttt{yzeng@math.tsinghua.edu.cn}}

\noindent
{\small {\it $^2$Department of Applied Mathematics,
  China Agricultural University,  Beijing 100083, P.R. China,
e-mail: }\texttt{xjliu@cau.edu.cn}}

\vskip 12pt { \small\noindent\bf Abstract.}
 {In this paper, we construct the
 bilinear identities for the wave functions of an extended Kadomtsev-Petviashvili (KP)
  hierarchy, which is the KP hierarchy with particular extended flows
  (2008, Phys. Lett. A, 372: 3819).
  By introducing an auxiliary
  parameter (denoted by $z$), whose flow corresponds to the so-called squared eigenfunction symmetry of KP
  hierarchy, we find the tau-function
  for this extended KP hierarchy.
  It is  shown that the bilinear
  identities will generate all the Hirota's bilinear equations for the zero-curvature forms of
  the extended KP hierarchy, which includes two types of KP equation with self-consistent sources (KPSCS).
  It seems that the Hirota's bilinear equations obtained in this paper for KPSCS
  are in a simpler form by comparing with the results by Hu and Wang (2007, Inverse Problems, 23: 1433). }\\

\noindent{\small\bf Mathematics Subject Classifications (2010):} 37K10.
\\

\noindent {\small\bf Key words:} KP hierarchy, bilinear identity,
tau-function, Hirota's bilinear forms.

\section{Introduction}
Sato theory has fundamental importance in the study of integrable systems (see \cite{DateMiwa1983}
and references therein). It reveals the infinite dimensional Grassmannian structure of space of
$\ta$-functions, where the $\ta$-functions are solutions for the Hirota's bilinear form of
Kadomtsev-Petviashvili (KP) hierarchy. The key point to this important discovery is a bilinear
residue identity for wave functions called \emph{bilinear identity}. Bilinear identity plays an
important role in the proof of existence for $\ta$-function.
 It also serves as the generating function
of the Hirota's bilinear equations for KP hierarchy \cite{Cheng1994,LorisWillox1997,ShenTu2011}.  In
this paper, we will construct the bilinear identity for the extended KP hierarchy \cite{Liu2008a}.

The study of integrable generalization is an important subject in the study of integrable
systems. Several approaches for the generalizations have been developed, e.g., constructing new
flows to \emph{extend} original systems. There are many ways to introduce new flows to make a new
compatible integrable system. In \cite{Xiong1992}, the KP hierarchy is extended by properly
combining additional flows. In \cite{Kamata2002}, extension of KP hierarchy is formulated by
introducing fractional order pseudo-differential operators. In \cite{Dimakis2004a,Dimakis2004b},
Dimakis and Muller-Hoissen extended the Moyal-deformed hierarchies by including additional evolution
equations with respect to the deformation parameters. In \cite{Carlet2003}, Carlet, Dubrovin and
Zhang defined a logarithm of the difference Lax operator and got the extended (2+1)D Toda lattice
hierarchy. Later, the Hirota's bilinear formalism and the relations of extended (2+1)D Toda lattice
hierarchy and extended 1D Toda lattice hierarchy have been
studied \cite{Milanov2007,Takasaki2010}. In
\cite{Li2010}, Li, {\it et al}, studied the $\ta$-functions and bilinear identities for the extended
bi-graded Toda lattice hierarchy.

Recently, we proposed a different and general approach to extend ($2+1$)-dimensional integrable
hierarchies by using the symmetries generating functions (or the squared eigenfunction symmetries)
\cite{Liu2008a}. This kind of extended KP hierarchy can be thought as the generalizations of the KP
hierarchy by squared eigenfunction symmetries (or \emph{ghost flows})
\cite{Aratyn1998,Oevel1993,Oevel1993a,Oevel1994}.  The reason for constructing such kind of extended
KP hierarchy is that it includes two types of KP equation with self-consistent sources (KPSCS-I and
KPSCS-II), which has important applications in hydrodynamics, plasma physics and solid state physics
(see, e.g., \cite{Doktorov-Shchesnovich-1995,Grinevich-Taimanov-2008,Hu2007,Lin2001,Lin2006,
  Melnikov1983,Melnikov1984,Melnikov1987,Melnikov1989,Zeng2000,Zhang2003}).
The extended KP
hierarchy also includes the well-known $k$-constrained KP hierarchy
\cite{Cheng1995,Konopelchenko1991,Konopelchenko1992} and several (1+1)-dimensional solition
equations with sources as reductions. It has been shown that many (2+1)-dimensional integrable
systems can be extended in this way \cite{Huang2011,Lin2008,Liu2008,Ma2010,Yao2009}.  We also
proposed a generalized dressing approach to construct \emph{Wronskian} type solutions (including
soliton solutions) for the extended hierarchies \cite{Huang2011,Lin2010,Liu2009,Liu2010,Yao2009}.

KPSCS-I and KPSCS-II can also be written in Hirota bilinear forms. In \cite{Hu2006,Wang2007}, Hu,
Wang, {\it et al}, proposed a \emph{source generation procedure} to construct the Hirota bilinear
form for KPSCS (I and II) on the basis of the bilinear form for the original KP equation.

However, it turns out that it is not an easy task to find Hirota bilinear form for KPSCS (I and
II). More generally, it remains unsolved to give the Hirota bilinear form for every
zero-curvature equation in the extended KP hierarchy. The most natural way to solve this problem is
to consider the bilinear identities of the extended KP hierarchy. Because of the importance of the
bilinear identities in Sato theory, the existence of bilinear identities could be an important open
question for the extended KP hierarchy.

In this paper, we will construct the bilinear identities for the extended KP hierarchy. Moreover, it will
serve as a generating function of all the Hirota bilinear forms for the zero-curvature equations in the
extended KP hierarchy. In particular, it generates Hirota bilinear forms for KPSCS (both type I and
II). The squared eigenfunction symmetry plays a crucial role in our construction.  It seems that the
Hirota's bilinear equations obtained in this paper for KPSCS are in a simpler form by comparing with
the results by Hu and Wang \cite{Hu2007}.

This paper is organized as follows. In Section \ref{sec:bi-KP-sq}, the bilinear identities for KP
hierarchy with a squared eigenfunction symmetry are constructed. In Section \ref{sec:bi-exKP}, by
making the squared eigenfunction symmetry as an auxiliary flow, we construct the bilinear identities
for the extended KP hierarchy. In addition, we prove any wave function satisfying bilinear
identities will be a wave function for the extended KP hierarchy. In Section \ref{sec:tau-exKPH}, we
find the $\ta$-function for the extended KP hierarchy. We also find the generation functions for
Hirota bilinear form for the extended KP hierarchy. Subsequently,
several examples including the Hirota
bilinear form for the KPSCS (I and II) and a higher order system in the hierarchy are
given. In Section \ref{sec:Back-nonlinear}, we show how to go back from Hirota bilinear forms to
nonlinear partial differential equations (PDEs) for KPSCS, which verifies the correctness of our
construction. In the last section, we will give conclusion and remarks, and discuss some problems
for further exploration.

\section{Bilinear Identities for KP hierarchy with the Squared Eigenfunction Symmetry}
\label{sec:bi-KP-sq}
Consider the system given in
\cite{Oevel1994}, defined by assuming that $L=\pa + \sum_{i=1}^\infty
u_i\pa^{-i}$ satisfies both the ordinary KP flows
\begin{subequations}\label{eqs:SE}
\begin{equation}
  \label{eq:KP}
  \pa_{t_n} L=[L^n_+,L],
\end{equation}
and a new $\pa_z$-flow given by
\begin{equation}
  \label{eq:sqs}
  \pa_z L = [q\pa^{-1} r, L].
\end{equation}
Here the Lax operator $L$ is a pseudo-differential operator (see \cite{dickey2003soliton} for a
detailed introduction). The ``$+$'' sign in subscript part of $L^n_+$ indicates the projection to
the non-negative part of $L^n$ with respect to the power of $\pa$. There are a series of time
variables $t_n$, and we can set $x=t_1$ according to (\ref{eq:KP}).

A sufficient condition for the commutativity of $\pa_z$- and $\pa_{t_n}$-flow is obtained by imposing
$q$ and $r$ to be an eigenfunction and an adjoint eigenfunction of KP hierarchy, respectively. Namely, we
should assume
\begin{align}
  \pa_{t_n} q &= L^n_+(q), \label{eq:eigenfun}\\
  \pa_{t_n} r &= - (L^n_+)^*(r), \label{eq:adj-eigenfun}
\end{align}
for all $n\in\mathbb{N}$. The $\pa_z$-flow describes a symmetry for KP hierarchy, which is called
the \emph{squared eigenfunction symmetry} \cite{Oevel1994} or a \emph{ghost flow}
\cite{Aratyn1998,Aratyn1999}.
\end{subequations}

It is well-known that dressing operator $W=\sum_{i\ge 0}w_i\pa^{-i}$ ($w_0=1$) plays a very
important role in KP theory. Here $w_i$ ($i>0$) are functions of all $t_n$ and $z$. Suppose the
$t_n$-evolution equations for $W$ are given by
\begin{equation}
  \pa_{t_n}W = -\left(W\pa^n W^{-1}\right)_- W, \qquad (n\in\mathbb{N}).\label{eqs:W-SE}
\end{equation}
Then Lax operator $L$ is related to the dressing operator $W$ by $L=W\pa W^{-1}$. We define the wave
and adjoint wave function as
\begin{subequations}\label{eqs:wave-func}
\begin{align}
  w(z,\bt,\ld) &= W e^{\xi(\bt,\ld)},\\
  w^*(z,\bt,\ld) &= (W^*)^{-1} e^{-\xi(\bt,\ld)}
\end{align}
\end{subequations}
where $\xi(\bt,\ld)=\sum_{i\ge 1}t_i\ld^i$ and $\bt=(t_1,t_2,...)$.
At this moment, we temporally ignore the parameter
$z$. Then the original KP theory states that a residue identity, called \emph{bilinear identity}
(\ref{eq:usualBI}), holds for (adjoint) wave functions:
\begin{equation}\label{eq:usualBI}
  \res_\ld\; w(z,\bt,\ld) \; w^*(z,\bt',\ld) =0,
\end{equation}
where $\btt=(t_1',t_2',...)$ and the residue of $\ld$ can be simply considered as the coefficient of
$\ld^{-1}$ in the Laurent expansion. Since we have ignored the parameter $z$ in $w$ and $w^*$, the
proof of (\ref{eq:usualBI}) is the same as the original one given in \cite{DateMiwa1983}. Notice
that the parameter $z$ takes the same value for $w$ and $w^*$ in (\ref{eq:usualBI}). The key step in
the proof of (\ref{eq:usualBI}) is an important lemma. We recall it here \cite{dickey2003soliton}:
\begin{lm}\label{lm:key}
  \begin{equation}
    \label{eq:lemma}
    \res_\pa \; P\cdot Q^* = \res_\ld \; P (e^{\xi(\bt,\ld)})\cdot Q (e^{-\xi(\bt,\ld)}),
  \end{equation}
  where $P$ and $Q$ are pseudo-differential operators. The residue with respect to $\pa$ on the left
  hand side of (\ref{eq:lemma}) is defined as the coefficient of $\pa^{-1}$ for a
  pseudo-differential operator.
\end{lm}
Actually, it is not difficult to prove that the bilinear identity (\ref{eq:usualBI}) is equivalent
to the KP hierarchy (\ref{eq:KP}) (see \cite{dickey2003soliton} for the proof). Namely, if we start
with (\ref{eq:KP}) or (\ref{eqs:W-SE}), we can prove that wave functions defined by
(\ref{eqs:wave-func}) will satisfy the bilinear identity (\ref{eq:usualBI}). Conversely, if we have
wave and adjoint wave functions satisfying bilinear identity (\ref{eq:usualBI}), then we can
immediately know that the dressing operator $W$ corresponding to the wave function $w$ will satisfy
(\ref{eqs:W-SE}), which makes the operator $L$ defined by dressing form $W\pa W^{-1}$ as a Lax operator
for KP hierarchy.

Then, if we consider KP hierarchy with squared eigenfunction symmetry,
we can see that
the dressing
operator should satisfy another equation:
\begin{equation}
  \label{eq:W-z}
  \pa_z W = q\pa^{-1} r W.
\end{equation}
In this case, we have:
\begin{thm}\label{thm:KP-SES}
  The KP hierarchy with a squared eigenfunction symmetry (\ref{eqs:SE}) is equivalent to the
  following residue identities
  \begin{subequations}\label{eqs:BI-SE}
    \begin{align}
      & \res_\ld\; w(z,\bt,\ld) \ \cdot\  w^*(z,\btt,\ld)=0, \label{eq:BI-ww*}\\
      & \res_\ld\; w_z(z,\bt,\ld) \ \cdot\  w^*(z,\btt,\ld)=q(z,\bt)r(z,\btt), \label{eq:BI-w_z-w*} \\
      & \res_\ld\; w(z,\bt,\ld)  \ \cdot\  \pa^{-1}\left(q(z,\btt) w^*(z,\btt,\ld)\right)=-q(z,\bt),\label{eq:BI-ztn-q}\\
      & \res_\ld\; \pa^{-1}\left(r(z,\bt) w(z,\bt,\ld)\right) \ \cdot\  w^*(z,\btt,\ld)=r(z,\btt), \label{eq:BI-ztn-r}
    \end{align}
  \end{subequations}
  where the inverse of $\pa$ is understood as pseudo-differential operator acting
  on an exponential function, e.g., $\pa^{-1}\left(r w\right)=(\pa^{-1}r W) (e^{\xi(\bt,\ld)})$.
\end{thm}
\begin{proof}
  Frist, we will prove the residue identities (\ref{eqs:BI-SE}) from (\ref{eqs:SE}),
  (\ref{eqs:W-SE}) and (\ref{eq:W-z}).  Equation (\ref{eq:BI-ww*}) is easy to prove, since $L$
  satisfies the evolution equations of original KP hierarchy.

  To prove (\ref{eq:BI-w_z-w*}), we only need to show
  \begin{displaymath}
    \res_\ld \; w_z(z,\bt,\ld) \, \cdot (\pa_{t_1}^{m_1}\cdots \pa_{t_k}^{m_k} w^*(z,\bt,\ld)) = q(z,\bt) \,
    \pa_{t_1}^{m_1}\cdots \pa_{t_k}^{m_k} r(z,\bt)
  \end{displaymath}
  for arbitrary $k\ge 1$ and $m_i\ge 0$. Here we have a key observation: the action of mixed partial
  derivatives $\pa_{t_1}^{m_1}\cdots\pa_{t_k}^{m_k}$ on $w^*$ can be written as $P_{m_1 \cdots m_k}
  w^*(z,\bt,\ld)$,
and
  \begin{displaymath}
    \pa_{t_1}^{m_1}\cdots\pa_{t_k}^{m_k}r(z,\bt) = P_{m_1\cdots m_k}\left(r(z,\bt)\right),
  \end{displaymath}
  where $P_{m_1\cdots m_k}$ is a differential operator in $\partial$. Then by noticing
  $w_z(z,\bt,\ld) = q \pa^{-1} r w(z,\bt,\ld)$ and Lemma \ref{lm:key}, we have
  \begin{align*}
    &\res_\ld \; w_z(z,\bt,\ld)\, \pa_{t_1}^{m_1}\cdots \pa_{t_k}^{m_k} w^*(z,\bt,\ld) \\
    =& \res_\ld \; q(z,\bt) \pa^{-1} r(z,\bt) w(z,\bt,\ld)\cdot P_{m_1\cdots m_k} w^*(z,\bt,\ld)\\
    =& \res_\pa \; q(z,\bt) \pa^{-1} r(z,\bt) P_{m_1\cdots m_k}^* \\
    =& q(z,\bt) P_{m_1\cdots m_k}(r(z,\bt)).
  \end{align*}
  Analogously, we have another residue identity:
  \begin{equation}\label{eq:dual-BI-exp}
    \res_\ld \; w(z,\bt,\ld)\, w_z^*(z,\btt,\ld) = q(z,\bt) r(z,\btt).
  \end{equation}
  However this equation is equivalent to (\ref{eq:BI-w_z-w*}), so this equation is not included in
  (\ref{eqs:BI-SE}).

  The residue identities (\ref{eq:BI-ztn-q}) and (\ref{eq:BI-ztn-r}) can be easily derived from
  (\ref{eq:BI-w_z-w*}) (and (\ref{eq:dual-BI-exp})) if one substitutes $w_z = q\pa^{-1}r w$ and
  $w_z^* = - r\pa^{-1} q w^*$ into (\ref{eq:BI-w_z-w*}) (or (\ref{eq:dual-BI-exp})) and eliminate $q$ or $r$.

The proof for the reverse part in this theorem is written as
the following proposition.
\end{proof}

\begin{prop}
  If there are functions $q(z,\bt)$, $r(z,\bt)$,
  \begin{displaymath}
    w(z,\bt,\ld) = \bigg(1+\sum_{i\ge 1}w_i\ld^{-i}\bigg) e^{\xi(\bt,\ld)} \quad\text{and}\quad
    w^*(z,\bt,\ld) = \bigg(1+\sum_{i\ge 1} w_i^*\ld^{-i}\bigg) e^{-\xi(\bt,\ld)}
  \end{displaymath}
  satisfying the identities (\ref{eqs:BI-SE}), then there exists a pseudo-differential operator $W$,
  and Lax operator $L:=W\pa W^{-1}$ such that $L$, $q$ and $r$ satisfy (\ref{eqs:SE}).
\end{prop}

\begin{proof}
  Supposing $W=1+\sum_{i\ge 1}w_i\pa^{-i}$, $\widetilde W = 1+\sum_{i\ge 1}w_i^*\pa^{-i}$, it is
  easy to see $w(z,\bt,\ld) = W (e^{\xi(\bt,\ld)})$ and $w^*(z,\bt,\ld) = \widetilde W
  (e^{-\xi(\bt,\ld)})$. Then from (\ref{eq:BI-ww*}) and Lemma \ref{lm:key}, we know the
  $\pa$-residue {$\res_\pa \; W\widetilde W^*\pa^m $} vanishes for $m\ge 1$, which implies the
  negative part $(W\widetilde W^*)_-$ is $0$. Furthermore, the non-negative part of $W\widetilde W^*$ is
  $1$. So we have proved that $W \widetilde W^*=1$, which means $\widetilde W = (W^*)^{-1}$.

  Now, let's define $L:=W\pa W^{-1}$. We will prove $W_{t_n}=-L^n_- W$, which implies (\ref{eq:KP}).
  From the definition
  form of $W$, it is easy to see that $(W_{t_n}+L^n_- W)_+=0$. Again from
  (\ref{eq:BI-ww*}) we have
  \begin{align*}
    & \res_\pa (W_{t_n}W^{-1}+L^n_-)_- \pa^m=\res_\pa (W_{t_n}W^{-1}+(W\pa^n W^{-1})_-)_- \pa^m\\
    & =\res_\ld W_{t_n} \mbox{e}^\xi \cdot(-\pa)^m (W^*)^{-1}\mbox{e}^{-\xi} +\res_\pa (W\pa^n
    W^{-1}-(W\pa^n
    W^{-1})_+)\pa^m \\
    & =\res_\ld W_{t_n} \mbox{e}^\xi \cdot (-\pa)^m (W^*)^{-1}\mbox{e}^{-\xi}+\res_\ld W\pa^n \mbox{e}^\xi \cdot
    (-\pa)^m
    {W^*}^{-1}\mbox{e}^{-\xi} \\
    & = \res_\ld w_{t_n}(z,\bt,\ld) (-\pa)^m w^*(z,\bt,\ld)=0, \quad \text{(for }m>0\text{)}
  \end{align*}
  so we obtain $W_{t_n}=-L^n_- W$.

  Next, we need to prove $W_{z}=q\pa^{-1}r W$, which implies (\ref{eq:sqs}). It is easy to see that
  $(W_{z}W^{-1})_+=0$, and by using (\ref{eq:BI-w_z-w*}) we have
    \begin{align*}
      &\res_\pa W_{z}W^{-1} \pa^m=\res_\ld W_z\mbox{e}^\xi \cdot (-\pa)^m {W^*}^{-1}\mbox{e}^{-\xi}\\
      &=\res_\ld w_z(z,\bt,\ld)\cdot (-\pa)^m w^*(z,\bt,\ld)=q(-\pa)^m r \quad \text{(for }m>0 \text{)},
  \end{align*}
  which means $W_zW^{-1}=\sum_{m=0}^\infty q(-\pa)^m(r) \pa^{-m-1}=q\pa^{-1}r$, namely,
  $W_{z}=q\pa^{-1}r W$.

  As the final step, we will prove (\ref{eq:eigenfun}) and (\ref{eq:adj-eigenfun}). We have already
  proved $W_{t_n}=-L^n_- W$, which implies $w_{t_n}=L^n_+(w)$. Then according to
  (\ref{eq:BI-ztn-q}), we find
  \begin{align*}
    q_{t_n}(z,\bt)&= -\res_\ld \; w_{t_n}(z,\bt,\ld) \; \cdot \pa^{-1}\left(q(z,\btt) w^*(z,\btt,\ld)\right)\\
    &= L^n_+\left(-\res_\ld w(z,\bt,\ld)\; \cdot \pa^{-1}\left(q(z,\btt) w^*(z,\btt,\ld)\right) \right)\\
    &=L^n_+(q(z,\bt)).
  \end{align*}
  In a similar way, we can find (\ref{eq:adj-eigenfun}) by using (\ref{eq:BI-ztn-r}).
\end{proof}

\section{Bilinear Identity for the Extended KP Hierarchy}
\label{sec:bi-exKP}
In \cite{Liu2008a}, we defined a new extended KP hierarchy. The key idea is to combine a specific
$k$-th order flow of KP hierarchy with the squared eigenfunction symmetry flow like
(\ref{eq:sqs}). In this way, we found two types of KP hierarchy with self-consistent
sources and their Lax representations.

In \cite{Hu2006,Hu2007,Wang2007}, the Hirota's bilinear equations for some (2+1)D integrable equations with sources are constructed by \emph{source
 generation procedure}. For the case of KP
equation, they found that these Hirota's bilinear equations correspond to the first and second type
of KP equation with self-consistent sources.

It is well-known that the bilinear identities are generating equations for Hirota bilinear equations
of KP hierarchy. The natural questions are how to find the bilinear identities for the extended KP
hierarchy and how to derive Hirota's bilinear equations from them. So in this section, we will give
a detailed construction on how to derive the bilinear identities for the extended KP hierarchy
(\ref{eqs:exKP}).

We remind that the extended KP hierarchy (for a fixed $k$)
\begin{subequations}\label{eqs:exKP}
  \begin{align}
    \pa_{t_n}L &= [L^n_+, L]\quad (n\neq k)\\
    \pa_{\bar{t}_k} L &= [L^k_+ + q\pa^{-1}r, L]\label{eq:tau_k-flow}\\
    \pa_{t_n} q &= L^n_+(q)\label{eq:eigenfun1}\\
    \pa_{t_n} r &= -(L^n_+)^*(r)\label{eq:adj-eigenfun1}
  \end{align}
\end{subequations}
is constructed by replacing an arbitrary fixed $k$-th flow $\pa_{t_k}$ by $\pa_{\bar{t}_k}$ where
the $\bar{t}_k$-flow is a combination of $\pa_{t_k}$-
 and $\pa_{z}$-flow given by
(\ref{eq:tau_k-flow}). Functions $q$ and $r$ are eigenfunction and adjoint eigenfunction satisfying
(\ref{eq:eigenfun1}) and (\ref{eq:adj-eigenfun1}), respectively.  Note that $q$ and $r$ depend on
$\bar{t}_k$ rather than $t_k$ in the extendeded KP hierarchy (\ref{eqs:exKP}).

\begin{rmk}
  In order to simplify the notions, we will still use the symbols $w(z,\bt,\ld)$, $w^*(z,\bt,\ld)$,
  $q(z,\bt)$, $r(z,\bt)$, $L$, $W$, etc., in this and the following sections, but they correspond to
  the case of the extended KP hierarchy (\ref{eqs:exKP}), for example, from now on,
  $\bt=(t_1,t_2,\cdots,t_{k-1},\bar{t}_k,t_{k+1},\cdots)$.
\end{rmk}

To find the bilinear identity for the extended KP hierarchy (\ref{eqs:exKP}), the key idea is
supposing that $L$ depends on the auxiliary variable $z$,
whose evolution  is
given by
(\ref{eq:sqs}).

Then we assume that the dressing operator for $L$ with auxiliary variable $z$ is given by $W =
1+\sum_{i\ge 1}w_i(z,\bt)\pa^{-i}$.
And the evolution of $W$ with respect to $z$
and $\bar{t}_k$ are given by (\ref{eqs:W-SE}) and
\begin{equation}
  \label{eq:W-tau_k}
  \pa_{\bar{t}_k} W = -L^k_- W + q\pa^{-1}r W.
\end{equation}
The definition forms for \emph{wave function} and \emph{adjoint wave function} are the same as
(\ref{eqs:wave-func}) except for the functions depending
on $\bar{t}_k$ rather than $t_k$ e.g.,
 $\xi(\bt,\ld) = \bar{t}_k\ld^k+\sum_{i\neq k}t_i\ld^i$.

Then we can prove the following theorem.

\begin{thm}
  The bilinear identity for the extended KP hierarchy (\ref{eqs:exKP})
  is given by the following sets of residue
  identities with auxiliary variable $z$:
  \begin{subequations}\label{eqs:BI-exKP}
    \begin{align}
      &\res_\ld \;w(z-\bar{t}_k,\bt,\ld)\; \cdot w^*(z-\bar{t}_k',\bt',\ld)=0,\label{eq:BI1-exKP}\\
      &\res_\ld \;w_z(z-\bar{t}_k,\bt,\ld)\;\cdot w^*(z-\bar{t}_k',\bt',\ld) = q(z-\bar{t}_k,\bt)
      r(z-\bar{t}_k',\bt'),\label{eq:BI2-exKP}\\
     &
   \res_\ld\; w(z-\bar{t}_k,\bt,\ld) \; \cdot
        \pa^{-1}\left(q(z-\bar{t}_k',\bt') w^*(z-\bar{t}_k',\bt',\ld)\right)
      = -q(z-\bar{t}_k,\bt),\label{eq:BI3-exKP} \\
      &\res_\ld\; \pa^{-1}\left(r(z-\bar{t}_k,\bt) w(z-\bar{t}_k,\bt,\ld)\right)\; \cdot
    w^*(z-\bar{t}_k',\bt',\ld) =
      r(z-\bar{t}_k',\bt'),\label{eq:BI4-exKP}
    \end{align}
  \end{subequations}
  where the inverse of $\pa$ is understood as pseudo-differential operator acting on an exponential
  function, e.g., $\pa^{-1}(r w)=(\pa^{-1}r W) (e^\xi)$.
\end{thm}

\begin{proof}
  We notice that
  $\frac{d}{d\bar{t}_k}w(z-\bar{t}_k,\bt,\ld)
  = w_{\bar{t}_k}(z-\bar{t}_k,\bt,\ld)-w_{z}(z-\bar{t}_k,\bt,\ld)
  =L^k_+(w)$. 
  The $\ld$-residue of $[\frac{d}{d\bar{t}_k}w(z-\bar{t}_k,\bt,\ld)] w^*(z-\bar{t}_k,\bt,\ld)$ 
  simply vanishes according to Lemma \ref{lm:key}. The same is true for the $\ld$-residues of
  arbitrary mixed derivatives $[\frac{d^{\sum
      m_i}}{dt_1^{m_1}dt_2^{m_2}\cdots}w(z-\bar{t}_k,\bt,\ld)] w^*(z-\bar{t}_k,\bt,\ld)$. 
  Then identity (\ref{eq:BI1-exKP}) holds.

  The proofs of (\ref{eq:BI2-exKP})-(\ref{eq:BI4-exKP}) are almost the same as the proofs of
  (\ref{eq:BI-w_z-w*})-(\ref{eq:BI-ztn-r}). One should notice again that
  $\frac{d}{d\bar{t}_k}w(z-\bar{t}_k,\bt,\ld)=L^k_+(w)$.
\end{proof}

\begin{thm}
  Suppose that there are (adjoint) wave functions
    $$ w(z,\bt,\ld)=\bigg(\sum_{i\ge 0}w_i(z,\bt)\ld^{-i}\bigg)e^{\xi(\bt,\ld)},
    \quad w^*(z,\bt,\ld) = \bigg(\sum_{i\ge 0}\tilde{w}_i(z,\bt)\ld^{-i}\bigg) e^{-\xi(\bt,\ld)}, $$
    (with $w_0=\tilde{w}_0=1$), $q(z,\bt)$ and $r(z,\bt)$ satisfying the bilinear relations
    (\ref{eqs:BI-exKP}), then the pseudo-differential operator $L = W\pa W^{-1}$, $W= \sum_{i\ge
      0}w_i(z,\bt)\pa^{-i}$, and functions $q$ and $r$ give a solution to the extended KP hierarchy
    (\ref{eqs:exKP}).
\end{thm}

\begin{proof}
  Let us define $\tilde{W}:= \sum_{i\ge 0}\tilde{w}_i(z,\bt)\pa^{-i}$. With the same argument as in
  Theorem \ref{thm:KP-SES}, we will find that $\tilde{W} = (W^{-1})^*$. Then from
  (\ref{eq:BI2-exKP}) and Lemma \ref{lm:key}, we find that for any $m>0$,
  \begin{displaymath}
    \res_\pa W_z W^{-1}(-\pa)^m = q \pa^m(r),
  \end{displaymath}
  which means $(W_z W^{-1})_- = q\pa^{-1}r$. Notice that the non-negative part of $W_z W^{-1}$
  vanishes, so we have $W_z = q\pa^{-1}r W$.

  Next, from the coefficient of Taylor expansion of (\ref{eq:BI1-exKP}), we find
  \begin{displaymath}
    \res_\ld \; \dfrac{d}{d\bar{t}_k}w(z-\bar{t}_k,\bt,\ld)\cdot \pa^mw^*(z-\bar{t}_k,\bt,\ld) = 0.
  \end{displaymath}
  By realizing that $\dfrac{d}{d\bar{t}_k}w(z-\bar{t}_k,\bt,\ld)$ is
  $(\dfrac{d}{d\bar{t}_k}W(z-\bar{t}_k,\bt)+L^kW) \exp(\xi(\bt,\ld))$ and using Lemma \ref{lm:key},
  we find
  \begin{displaymath}
    \res_\pa (\dfrac{d}{d\bar{t}_k}W(z-\bar{t}_k,\bt) + L^k_-W)W^{-1}(-\pa)^m = 0,\quad \text{(for } m>0\text{)}
  \end{displaymath}
  which means $\dfrac{d}{d\bar{t}_k}W(z-\bar{t}_k,\bt) = - L^k_- W$. Hence $\pa_{\bar{t}_k}W(z,\bt) = (-L^k_- +
  q\pa^{-1}r)W$.

  For equations (\ref{eq:eigenfun1}) and (\ref{eq:adj-eigenfun1}), the proof can be done by
  differentiating (\ref{eq:BI3-exKP}) and (\ref{eq:BI4-exKP}) directly. Since $\pa_{t_n}w =
  L^n_+(w)$, $\pa_{t_n}w^* = -(L^n_+)^*(w^*)$, the rest of proof is obvious.
\end{proof}

\section{Tau-Function for the Extended KP Hierarchy}
\label{sec:tau-exKPH}
The existence of $\tau$-function for original bilinear identity of KP hierarchy is proved in
\cite{DateMiwa1983}. In our case, the wave function $w(z-\bar{t}_k,\bt,\ld)$ and
$w^*(z-\bar{t}_k,\bt,\ld)$ satisfy exactly the same bilinear identity (\ref{eq:BI1-exKP}) as the
original KP case
if one considers $z$ as an additional parameter. So it is reasonable to
assume the existence of $\tau$-function and make the following ansatz:
\begin{subequations}\label{eqs:ansatz}
\begin{align}
  &
  w(z-\bar{t}_k,\bt,\ld)=\frac{\ta(z-\bar{t}_k+\frac{1}{k\ld^k},\bt-[\ld])}{\ta(z-\bar{t}_k,\bt)}\cdot
  \exp\xi(\bt,\ld),\\
  & w^*(z-\bar{t}_k,\bt,\ld)=\frac{\ta(z-\bar{t}_k-\frac{1}{k\ld^k},\bt+[\ld])}{\ta(z-\bar{t}_k,\bt)}\cdot \exp
  (-\xi(\bt,\ld)),
\end{align}
where $[\ld]=\left(\frac{1}{\ld},\frac1{2\ld^2},\frac1{3\ld^3},\cdots\right)$.  According to
\cite{Cheng1994}, we should make further assumptions:
\begin{align}
  q(z,\bt)=\frac{\sg(z,\bt)}{\ta(z,\bt)}, \qquad r(z,\bt)=\frac{\rho(z,\bt)}{\ta(z,\bt)}.
\end{align}
Then, similar with \cite{Cheng1994}, we have the following results.
\begin{align}
  &\pa^{-1}(r(z-\bar{t}_k,\bt)w(z-\bar{t}_k,\bt,\ld)) = \dfrac{\rho(z-\bar{t}_k+\frac{1}{k\ld^k},\bt-[\ld])}{\ld
    \ta(z-\bar{t}_k,\bt)}e^{\xi(\bt,\ld)},\\
  &\pa^{-1}(q(z-\bar{t}_k,\bt)w^*(z-\bar{t}_k,\bt,\ld)) = \dfrac{-\sg(z-\bar{t}_k-\frac{1}{k\ld^k},\bt+[\ld])}{\ld
  \ta(z-\bar{t}_k,\bt)}e^{-\xi(\bt,\ld)}.
\end{align}
\end{subequations}

To find Hirota's bilinear equations for extended KP hierarchy (\ref{eqs:exKP}), we substitute
(\ref{eqs:ansatz}) into (\ref{eqs:BI-exKP}), and get
\begin{subequations}\label{eqs:exBI-tau}
\allowdisplaybreaks
\begin{align}
  &\res_\ld\;\bar{\ta}\big(z,\bt-[\ld]\big)\;\bar{\ta}\big(z,\bt'+[\ld]\big)
  \;e^{\xi(\bt-\bt',\ld)} =0,\\
  &\res_\ld\;\bar\ta_z\big(z,\bt-[\ld]\big)\;
    \bar\ta\big(z,\bt'+[\ld]\big)\;e^{\xi(\bt-\bt',\ld)}\notag\\
  &\phantom{==}-
  \res_\ld\;\bar\ta\big(z,\bt-[\ld]\big)\;(\pa_z\log\bar\ta(z,\bt))
    \ \bar\ta\big(z,\bt'+[\ld]\big)\;
  e^{\xi(\bt-\bt',\ld)} \notag \\
  &\phantom{===========================}
  = \bar\sg(z,\bt)\ \bar\rho(z,\bt'),\\
  &\res_\ld\;\ld^{-1}\bar\ta\big(z,\bt-[\ld]\big)\;
  \bar\sg\big(z,\bt'+[\ld]\big)\; e^{\xi(\bt-\bt',\ld)}= \bar\sg(z,\bt)\ \bar\ta(z,\bt'),\\
  &\res_\ld\;\ld^{-1}\bar\rho\big(z,\bt-[\ld]\big) \
  \bar\ta\big(z,\bt'+[\ld]\big)\; e^{\xi(\bt-\bt',\ld)} =  \bar\rho(z,\bt')\ \bar\ta(z,\bt).
\end{align}
Here the bar $\bar{\phantom{g}}$ over a function $f(z,\bt)$ is defined as
$\bar f(z,\bt)\equiv f(z-\bar{t}_k,\bt)$, e.g,
$\bar{\ta}\big(z,\bt-[\ld]\big)\equiv
{\ta}\big(z-(\bar{t}_k-\frac{1}{k\ld^k}),\bt-[\ld]\big)$,
$\bar{\ta}\big(z,\bt'+[\ld]\big)\equiv
{\ta}\big(z-(\bar{t}_k+\frac{1}{k\ld^k}),\bt'+[\ld]\big)$.

\end{subequations}
After setting $\bt$ as $\bt+\by$ and $\btt$ as $\bt-\by$ in (\ref{eqs:exBI-tau}) with
$\by=(y_1,y_2,\cdots)$, We can write (\ref{eqs:exBI-tau}) as the following systems with Hirota
bilinear derivatives $\tilde D$ and $D_i$'s:
\begin{subequations}\label{gen-hirota}
  \begin{align}
    &\sum_{i\ge0}p_i(2\by)p_{i+1}(-\tilde{D})\exp{\left(\sum_{j\ge1} y_j D_j\right)}
    \bar \ta(z,\bt)\cdot \bar \ta(z,\bt)=0,\label{gen-hirota-tau}\\
    & \sum_{i\ge0}p_i(2\by)p_{i+1}(-\tilde{D})\exp\left(\sum_{j\ge1}y_j D_j\right)
    \bar\ta_{z}(z,\bt)\cdot\bar\ta(z,\bt)\notag\\
    &\phantom{=}
  -\sum_{i\ge0}p_i(2\by)\left(\pa_z\log\bar\ta(z,\bt+\by)\right)p_{i+1}(-\tilde{D})\bar\ta(z,\bt+\by)\cdot\bar\ta(z,\bt-\by) \notag\\
    &\phantom{===================}=\exp\left(\sum_{j\ge1}y_j D_j\right)\bar\sg(z,\bt)\cdot\bar\rho(z,\bt)\label{gen-hirota-tau-sigma-rho}\\
    &\sum_{i\ge0}p_i(2\by) p_i(-\tilde{D})\exp\left(\sum_{j\ge1}y_j
      D_j\right)\bar\ta(z,\bt)\cdot\bar\sg(z,\bt)=\exp\left(\sum_{j\ge1}y_jD_j\right)\bar\sg(z,\bt)\cdot\bar\ta(z,\bt),
    \label{gen-hirota-sigma}\\
    &\sum_{i\ge0}p_i(2\by) p_i(-\tilde{D})\exp\left(\sum_{j\ge1}y_j
      D_j\right)\bar\rho(z,\bt)\cdot\bar\ta(z,\bt)=\exp\left(\sum_{j\ge1}y_jD_j\right)\bar\ta(z,\bt)\cdot\bar\rho(z,\bt),
    \label{gen-hirota-rho}
  \end{align}
\end{subequations}
where
$\tilde{D}=(D_1,\frac12 D_2,\frac13
D_3,\cdots)$, $D_i$ is the well-known Hirota bilinear derivative $D_if\cdot g = f_{t_i}g-f g_{t_i}$,
and  $p_i(\by)$
 is the $i$-th \emph{Schur polynomial}, whose generating function is given by
\begin{displaymath}
  \exp\sum_{i=1}^\infty y_i\ld^i = \sum_{i=0}^\infty p_i(\by)\ld^i.
\end{displaymath}

\begin{rmk}
  The zero-th order
   term in (\ref{gen-hirota-tau-sigma-rho}) (with $y_j=0$ $\forall\ j$) can be written
  in the following forms with Hirota's operator
    $$ \sg\rho + D_x \tau_{z}\cdot  \tau
  = \sg\rho + D_{z} \tau_x \cdot \tau
  =  \sg\rho + \frac12 D_{x}D_{z} \tau \cdot \tau =0.$$
\end{rmk}

Note that (\ref{gen-hirota}) is the generating equations of Hirota's bilinear equations for the
extended KP hierarchy (\ref{eqs:exKP}), where the dependence of auxiliary parameter $z$ can be
 eliminated if one wants to convert bilinear equations to PDEs. In the following sections
we will give several examples to show 
the Hirota's bilinear forms
for the nonlinear evolution equations in extended KP hierarchy and how it can be transformed back to
nonlinear PDEs (which will correspond to the well-known KPSCS-I and II, etc.).

\begin{eg}[First type of KP equation with a source (KPSCS-I)
  \cite{Hu2007,Liu2008a,Melnikov1983,Melnikov1987}, i.e., the extended KP hierarchy (\ref{eqs:exKP})
  for $n=2$ and $k=3$]\label{eg:KP1}
  The Hirota equations for the KPSCS-I can be obtained as
  \begin{align*}
    & D_{x} \ta_z \cdot \ta +\sg\rho=0, & \text{from (\ref{gen-hirota-tau-sigma-rho}) with $y_j=0$} \\
    & (D_x^4+3D_{t_2}^2-4D_x(D_{\bar{t}_3}-D_z)) \ta \cdot \ta =0, & \text{from
      (\ref{gen-hirota-tau}) in $y_3$}\\
    & (D_{t_2}+D_x^2) \ta\cdot \sg =0, & \text{from (\ref{gen-hirota-sigma}) in $y_2$} \\
    & (D_{t_2}+D_x^2) \rho \cdot \ta=0, & \text{from (\ref{gen-hirota-rho}) in $y_2$}
  \end{align*}
  where $D_z$ is Hirota's derivative, i.e., $D_z f(z)\cdot g(z)=f_z g- f g_z$. Note that from the
  definition of $\bar\tau$, we know that $D_{\bar{t}_3}\bar\tau\cdot\bar\tau =
  (D_{\bar{t}_3}-D_z)\tau\cdot\tau$, which interprets the appearance of this term in the second
  equation.

\end{eg}

\begin{eg}[Second type of KP with a source (KPSCS-II) \cite{Hu2007,Liu2008a,Melnikov1983,Melnikov1987},
i.e., the extended KP hierarchy (\ref{eqs:exKP}) for $n=3$ and $k=2$]\label{eg:KP2}
  The Hirota equations for the KPSCS-II can be obtained as
  \begin{align*}
    & D_{x} \ta_z \cdot \ta +\sg\rho=0, & \text{from (\ref{gen-hirota-tau-sigma-rho}) with $y_j=0$} \\
    & (D_x^4+3(D_{\bar{t}_2}-D_z)^2-4D_xD_{t_3}) \ta \cdot \ta =0, & \text{from
      (\ref{gen-hirota-tau}) in $y_3$}\\
    & ((D_{\bar{t}_2}-D_z)+D_x^2) \ta\cdot \sg =0, & \text{from (\ref{gen-hirota-sigma}) in $y_2$} \\
    & ((D_{\bar{t}_2}-D_z)+D_x^2) \rho \cdot \ta=0. & \text{from
      (\ref{gen-hirota-rho}) in $y_2$} \\
    & (4D_{t_3}-D_x^3+3D_x(D_{\bar{t}_2}-D_z)) \ta \cdot \sg =0, &
    \text{from (\ref{gen-hirota-sigma}) in $y_3$}\\
    & (4D_{t_3}-D_x^3+3D_x(D_{\bar{t}_2}-D_z))\rho\cdot \ta = 0. & \text{from (\ref{gen-hirota-rho}) in $y_3$}
  \end{align*}
  It seems that the Hirota bilinear equations obtained here for KPSCS-II are in a simpler form
  comparing with the results by Hu and Wang \cite{Hu2007}.

\end{eg}

\begin{eg}[The extended KP hierarchy (\ref{eqs:exKP}) for $n=4$ and $k=2$ \cite{Liu2008a}]\label{eg:nextKP}
  The Hirota bilinear equations for this system can be obtained as
  \begin{align*}
    & D_{x} \ta_z \cdot \ta +\sg\rho=0, & \text{ from (\ref{gen-hirota-tau-sigma-rho}) with $y_j=0$} \\
    & \left(D_x^4 + 3 (D_{\bar{t}_2}-D_z)^2 - 4 D_x D_{t_3}\right) {\ta}\cdot {\ta} = 0, &
    \text{from (\ref{gen-hirota-tau}) in $y_3$ }\\
    & \left(3D_xD_{t_4} - 2 (D_{\bar{t}_2}-D_z)D_{t_3}-D_x^3 (D_{\bar{t}_2}-D_z)\right){\ta}\cdot {\ta}
    = 0,
    & \text{from (\ref{gen-hirota-tau}) in $y_4$}\\
    & \left((D_{\bar{t}_2}-D_z) + D_x^2\right)\ta\cdot \sg = 0, & \text{from (\ref{gen-hirota-sigma}) in $y_2$}\\
    & \left((D_{\bar{t}_2}-D_z) + D_x^2\right)\rho\cdot \ta = 0, &\text{from (\ref{gen-hirota-rho}) in $y_2$}\\
    & \left(4D_{t_3} - D_x^3 + 3D_x(D_{\bar{t}_2}-D_z)\right)\ta\cdot\sg = 0, &
    \text{from (\ref{gen-hirota-sigma}) in $y_3$}\\
    & \left(4D_{t_3} - D_x^3 + 3D_x(D_{\bar{t}_2}-D_z)\right)\rho\cdot\ta = 0, &\text{ from (\ref{gen-hirota-rho}) in $y_3$}\\
    & \left(18D_{t_4}+D_x^4 -6D_x^2 (D_{\bar{t}_2}-D_z)+3(D_{\bar{t}_2}-D_z)^2 + 8 D_x
      D_{t_3}\right)\ta\cdot\sg
    =0, &\text{from (\ref{gen-hirota-sigma}) in $y_4$}\\
    &\left(18D_{t_4} +D_x^4 -6D_x^2 (D_{\bar{t}_2}-D_z)+3(D_{\bar{t}_2}-D_z)^2 + 8 D_x
      D_{t_3}\right)\rho\cdot\ta =0. &\text{from (\ref{gen-hirota-rho}) in $y_4$}
  \end{align*}

\end{eg}

\section{Back to nonlinear equations from Hirota's bilinear equations}
\label{sec:Back-nonlinear}

To see the nonlinear PDEs corresponding to the Hirota's bilinear equations
in the previous examples
in last section, we convert the bilinear equations back to nonlinear PDEs in this section.
It is not as simple as generating special bilinar equations. We follow the way given in \cite{Hirota2004}.

Consider the following identities
\begin{subequations}\label{eqs:D-p}
  \begin{align}
    &\exp\left(\sum_i\delta_i D_i\right) \rho\cdot \tau =
    e^{2\cosh\left(\sum_i\delta_i\partial_i\right)\log\tau}\cdot
    e^{\sum_i\delta_i\pa_i}\left(\rho/\tau\right),\label{eq:D-p1}\\
    &\cosh\left(\sum_i\delta_iD_i\right)\tau\cdot\tau = e^{2\cosh(\sum_i\delta_i\pa_i)\log\tau}.\label{eq:D-p2}
  \end{align}
\end{subequations}
\begin{rmk}
  Note that in this section, $D_if\cdot g \defeq f_{t_i}g-f g_{t_i}$.
\end{rmk}

The identities (\ref{eqs:D-p}) are easy to prove. One should notice the definitions of operator
$D_i$, for example, (\ref{eq:D-p1}) is proved by checking from left hand side:
\begin{align*}
  \text{l.h.s.} &= \rho(t_1+\delta_1,t_2+\delta_2,\cdots ) \tau (t_1-\delta_1,t_2-\delta_2,\cdots),\\
  \text{r.h.s.} &=\exp\left[(e^{\sum_i\delta_i\pa_i}+e^{-\sum_i\delta_i\pa_i})\log\tau \right]\cdot
  \frac{\rho(t_1+\delta_1,t_2+\delta_2,\cdots)}{\tau(t_1+\delta_1,t_2+\delta_2,\cdots)}\\
  &=\exp\left[\log\tau(t_1+\delta_1,t_2+\delta_2,\cdots)+\log\tau(t_1-\delta_1,t_2-\delta_2,\cdots)\right]
  \\
  &\quad \cdot \frac{\rho(t_1+\delta_1,t_2+\delta_2,\cdots)}{\tau(t_1+\delta_1,t_2+\delta_2,\cdots)}\\
  &=\rho(t_1+\delta_1,t_2+\delta_2,\cdots ) \tau (t_1-\delta_1,t_2-\delta_2,\cdots).
\end{align*}

Using the transformations $u\defeq \pa_x^2\log\tau$ ($x\equiv t_1$), $r\defeq \rho/\tau$, $q \defeq \sigma/\tau$,
we can rewrite (\ref{eqs:D-p}) into the following form.
\begin{subequations}
  \begin{align}
    &\frac1{\tau^2}\sum_{n=0}^\infty\frac{\left(\sum_i\delta_iD_i\right)^n}{n!}\rho\cdot\tau =
    \exp\left[2\sum_{n=0}^\infty\frac{\sum_i\left(\delta_i\pa_i\right)^{2n}}{2n!}\pa^{-2}u\right]
    \cdot e^{\sum_i\delta_i\pa_i}r,\\
    &\frac1{\tau^2}\sum_{n=0}^\infty\frac{(\sum_i\delta_iD_i)^{2n}}{2n!}\tau\cdot\tau=
    \exp\left[2\sum_{n=0}^\infty\frac{\left(\sum_i\delta_i\pa_i\right)^{2n}}{2n!}\pa^{-2}u\right].
  \end{align}
\end{subequations}
Thus, relations between the Hirota bilinear derivatives and usual partial derivatives with respect
to original variables are established. We list some of these relations as follows.
\begin{center}
  \begin{tabular}[H]{|c||c||c||c|}
    \hline
    $\dfrac{D_1^2 \tau\cdot\tau}{\ta^2}$ &   $\dfrac{D_1 D_2 \tau\cdot\tau}{\ta^2}$ &
    $\dfrac{D_1D_3 \tau\cdot\tau}{\ta^2}$ &  $\dfrac{D_1D_4 \tau\cdot\tau}{\ta^2}$ \\
    \hline
    $2u$ & $2\pa^{-1}u_2$  & $2\pa^{-1}u_3$ & $2\pa^{-1}u_4$\\
    \hline
    \hline
    $\dfrac{D_{2}D_{3}\tau\cdot\tau}{\ta^2}$ &  $\dfrac{D_1^3 D_{2} \tau\cdot\tau}{\ta^2}$ &
    $\dfrac{D_1^4 \tau\cdot\tau}{\ta^2}$ & $\cdots$\\
    \hline
    $2\pa^{-2}u_{2,3}$ & $2u_{1,2} + 12u(\pa^{-1}u_{2})$ & $2u_{1,1}+12u^2$ & $\cdots$\\
    \hline
    \hline
    $\dfrac{D_1\rho\cdot\tau}{\tau^2}$ & $\dfrac{D_2\rho\cdot\tau}{\tau^2}$ &
    $\dfrac{D_3\rho\cdot\tau}{\tau^2}$ & $\dfrac{D_4\rho\cdot\tau}{\tau^2}$ \\
    \hline
    $r_1$ & $r_2$ & $r_3$ & $r_4$\\
    \hline
    \hline
    $\dfrac{D_1^2\rho\cdot\tau}{\tau^2}$ & $\dfrac{D_1D_2\rho\cdot\tau}{\tau^2}$
    & $\dfrac{D_1 D_3\rho\cdot\tau}{\tau^2}$  & $\dfrac{D_2^2\rho\cdot\tau}{\tau^2}$ \\
    \hline
    $r_{1,1} + 2ur$ & $r_{1,2} + 2r \pa^{-1}u_2$ &  $r_{1,3}+2r\pa^{-1}u_3$ & $r_{2,2} + 2r\pa^{-2}u_{2,2}$ \\
    \hline
    \hline
    $\dfrac{D_1^3\rho\cdot\tau}{\tau^2}$ &
    $\dfrac{D_1^2D_2\rho\cdot\tau}{\tau^2}$ &
    $\dfrac{D_1^4\rho\cdot\tau}{\tau^2}$ & $\cdots$\\
    \hline
    $r_{1,1,1} + 6ur_1$ & $r_{1,1,2} + 2ur_2 + 4r_1\pa^{-1}u_2$ &
    \begin{minipage}{1.1in}
      \begin{equation*}
        \begin{aligned}
          &r_{1,1,1,1} + 12 u r_{1,1}\\
          &2(6u^2 + u_{1,1})r
        \end{aligned}
      \end{equation*}
    \end{minipage} &$\cdots$
    \\
    \hline
  \end{tabular}
\end{center}
where the subscript $i$, $j$ $\cdots$ of $u_{i,j,\ldots}$ and $r_{i,j,\ldots}$ denote the derivative
with respect to the corresponding variables $t_i$, $t_j$ $\cdots$.

\begin{eg}
  Example \ref{eg:KP1} can be translated back to nonlinear PDEs as
  ($y\equiv t_2$)
  \begin{align*}
    &\pa_z\pa^{-1}  u +  q r=0,\\
    &(u_{xxx}+12u u_x -4u_{\bar{t}_3}+4u_z)_x + 3u_{yy}=0,\\
    & q_y = q_{xx} + 2uq,\\
    & r_y = -r_{xx} - 2u r.
  \end{align*}
  If we eliminate auxiliary variable $z$ from the above equations,
  we get the first type of KP equation with a source
  \cite{Hu2007,Liu2008a,Melnikov1983,Melnikov1987}
  \begin{align*}
      &  (4u_{\bar{t}_3}-u_{xxx} -12u u_x)_x -3 u_{yy} + 4(qr)_{xx}=0,\\
    & q_y = q_{xx} + 2uq,\\
    & r_y = -r_{xx} - 2u r,
  \end{align*}
  whose soliton solutions can be obtained by dressing method \cite{Liu2009} and Hirota method
  \cite{Hu2007}.
\end{eg}

\begin{eg}
  Example \ref{eg:KP2} can be translated back to nonlinear PDEs as
  ($y\equiv t_2, t\equiv t_3$)
  \begin{align*}
    &\pa_z\pa^{-1}  u +  q r=0,\\
    &(u_{xxx}+12 u u_x -4u_{t_3})_x + 3 (\pa_y-\pa_z)^2u=0,\\
    &q_y-q_z=q_{xx} + 2uq,\\
    &r_y-r_z = -r_{xx} - 2ur,\\
    &-4q_t + q_{xxx} + 6uq_x +3q_{xy}-3q_{xz} + 6q\pa_y\pa^{-1}u - 6q\pa_z\pa^{-1}u =0,\\
    &4r_t - r_{xxx} - 6ur_x +3r_{xy}-3r_{xz} + 6r\pa_y\pa^{-1}u - 6r\pa_z\pa^{-1}u =0.
  \end{align*}
  By eliminating auxiliary variable $z$, we reach
  the second type of KP equation with a source
  \cite{Hu2007,Liu2008a,Melnikov1983,Melnikov1987}
  \begin{align*}
    &(4u_t -u_{xxx} -12u u_x)_x -3u_{yy}=3\left[(qr)_y+(q_{xx}r-qr_{xx})\right]_x,\\
    &q_t = q_{xxx} + 3uq_x + \frac32 u_x q + \frac32q\pa^{-1}u_y + \frac32 q^2r,\\
    &r_t = r_{xxx} + 3ur_x + \frac32 u_x r - \frac32r\pa^{-1}u_y-\frac32 r^2q,
  \end{align*}
  whose soliton solutions can be obtained by dressing method \cite{Liu2009} and Hirota method
  \cite{Hu2007}.
\end{eg}

\begin{eg}
  Example \ref{eg:nextKP} can be translated to nonlinear PDEs as ($y\equiv t_2$)
  \begin{align*}
    &\pa_z\pa^{-1} u + q r=0,\\
    &2u_{xx}+12u^2+6\pa^{-2}[(\pa_y-\pa_z)^2u] - 8\pa^{-1}u_{t_3} = 0,\\
    &6\pa^{-1}u_{t_4}-4\pa^{-2}(u_{t_3,y}-u_{t_3,z})-12u\pa^{-1}(u_y-u_z) -
    2(u_{xy}-u_{xz})= 0,\\
    &q_y=q_z + q_{xx} + 2uq,\\
    &r_y=r_z - r_{xx} - 2ur,\\
    &4q_{t_3}=q_{xxx} + 6uq_x + 3\left(q_{xy}-q_{xz} + 2q\pa^{-1}(u_y-u_z)\right),\\
    &4r_{t_3}=r_{xxx} + 6ur_x - 3\left(r_{xy}-r_{xz} + 2r\pa^{-1}(u_y-u_z)\right),\\
    &18q_{t_4} =q_{xxxx}+ 12 u q_{xx} + 24 q_x \pa^{-1}(u_y-u_z) +\big(12 u^2 +
    2u_{xx}+16\pa^{-1}u_{t_3}\\
    &\phantom{18q_{t_4} =} + 6 \pa^{-2}(\pa_y-\pa_z)^2u\big)q + 12 u (q_y-q_z) + 6
    (q_{xxy}-q_{xxz}) + 3(\pa_y-\pa_z)^2q + 8q_{x t_3},\\
    &18r_{t_4} =-r_{xxxx}- 12 u r_{xx}+ 24 r_x\pa^{-1}(u_y-u_z) -\big(12 u^2  + 2u_{xx}+16
    \pa^{-1}u_{t_3} \\
    &\phantom{18r_{t_4}=}+6\pa^{-2}(\pa_y-\pa_z)^2u\big)r + 12 u (r_y-r_z) +6
    (r_{xxy}-r_{xxz}) - 3(\pa_y-\pa_z)^2r - 8 r_{x t_3}
  \end{align*}
  Eliminating the $\pa_z$ and $\pa_{t_3}$ terms, we get
  the extended KP hierarchy (\ref{eqs:exKP}) for $n=4$ and $k=2$ \cite{Liu2008a}
  \begin{subequations}\label{eqs:Hirota-nonlinear-k4m2}
    \begin{align}
      2u_{t_4} =& \pa^{-2}u_{yyy} + u_{xxy} +  4(u^2)_y + 4(\pa^{-1}u_y)u_x + 8(uqr)_x +\pa^{-1}(qr)_{yy} \notag \\
      &+ 2(q_{xx}r+qr_{xx})_x+(rq_x-qr_x)_y, \\
      q_{t_4} =& q_{xxxx}+ 4uq_{xx} +(4u_x +4qr+2\pa^{-1}u_y)q_x + (\pa^{-2}u_{yy} + \pa^{-1}(qr)_y \notag\\
      &+4u^2 + 2u_{xx}+u_y)q,\\
      -r_{t_4}  =& r_{xxxx}+4ur_{xx} + (4u_x - 4 qr -2\pa^{-1}u_y)r_x
      +(\pa^{-2}u_{yy}+\pa^{-1}(qr)_y \notag\\
      &+ 4u^2 + 2u_{xx} -u_y)r,
    \end{align}
  \end{subequations}
  whose Wronskian type solution (including soliton solutions) can be obtained by dressing method
  \cite{Liu2009}.  The correctness of (\ref{eqs:Hirota-nonlinear-k4m2}) can be verified from the
  zero-curvature equation of extended KP hierarchy (\ref{eqs:exKP}) with $n=4$ and $k=2$:
  \begin{subequations}\label{eqs:zc-k4m2}
    \begin{align}
      &B_{4,\bar{t}_2}-B_{2,t_4} +[B_4,B_2+q\pa^{-1} r]_+=0,\\
      &\pa_{t_4}q =B_4(q ), \\
      &\pa_{t_4}r =-B_4^*(r ).
    \end{align}
  \end{subequations}
  In the expressions of (\ref{eqs:zc-k4m2}), if we eliminate $u_2$ and $u_3$ from $\pa^0$ term by
  using $\pa^2$ and $\pa$ terms, we can verify that (\ref{eqs:zc-k4m2}) is exactly
  (\ref{eqs:Hirota-nonlinear-k4m2}).
\end{eg}

\section{Conclusion and Discussions}
\label{sec:conclusion}
In this paper, we constructed the bilinear identities (\ref{eqs:BI-exKP}) for the extended KP
hierarchy (\ref{eqs:exKP}) defined in \cite{Liu2008a}.  The bilinear identities (\ref{eqs:BI-exKP})
are used to derive the Hirota's bilinear equations (\ref{gen-hirota}) for all the zero-curvature
forms of the extended KP hierarchy (\ref{eqs:exKP}).  We have shown that the Hirota's bilinear forms
in Example \ref{eg:KP1} and Example \ref{eg:KP2} correspond to the KP equation with a self-consistent source
(KPSCS-I and II).  After
translating the Hirota equations back to nonlinear PDEs, the correctness of
our bilinear forms (\ref{eqs:BI-exKP}) are verified.  Another forms of Hirota's bilinear equations
for KPSCS-I and KPSCS-II have been given in \cite{Hu2007} by using \emph{source generation procedure},
where the auxiliary functions $P_i$ and $Q_i$ are introduced.  In this paper, by introducing an
auxiliary flow ($\partial_z$-flow), the bilinear identity for the whole extended KP hierarchy and a
simpler Hirota's form are obtained. To show the validity of our method, we gave an extra example in
the extended KP hierarchy (Example \ref{eg:nextKP}).

There are some important applications for bilinear identities of extended KP hierarchy.
As we know,
the quasi-periodic solutions for KP hierarchy can be constructed by using method
in algebraic geometry, where the construction of wave functions (or Baker-Akhiezer functions as in
quasi-periodic cases) are intimately related with bilinear identities, Riemann surfaces and divisors
on it.  It is very interesting to consider the quasi-periodic solutions for the extended KP
hierarchy (\ref{eqs:exKP}) when bilinear identities have been obtained in this paper. Another
interesting problem is to consider the bilinear identities for other extended hierarchies, such as
BKP, CKP,
2D Toda and discrete KP, etc. We would like to investigate these problems in the
future.

\section*{Acknowledgement}
The authors  would like to thank
R. Conte and M. Ma\~nas for valuable discussions.  This work is supported by National
Natural Science Foundation of China (10901090, 11171175,  11201477) and Chinese Universities Scientific Fund
(2011JS041).
RL is supported in part by the joint project between
Chinese NSFC and French CNRS in 2011.

\bibliographystyle{hunsrt}

\end{document}